\documentclass[11pt]{article}

\RequirePackage{amsmath}
\RequirePackage{amssymb}
\RequirePackage{amsthm}
\RequirePackage{epsfig}
\RequirePackage{mathrsfs}
\usepackage{amsfonts}
\usepackage{array}
\usepackage{enumerate}
\usepackage{graphicx}
\usepackage{longtable}
\usepackage{amsmath}
\usepackage{latexsym}
\usepackage{fancyhdr}
\usepackage[hidelinks]{hyperref}

\fancyheadoffset{2 pc}
\fancyfootoffset{2 pc}
\hypersetup{
 colorlinks = true,
 urlcolor = blue,
 linkcolor = blue,
 citecolor = red
}
\fancypagestyle{plain}{%
\lhead{}
\chead{\textcolor{red}{\footnotesize{\textit{International Journal of Mathematics Trends and Technology - Volume 7 Number 2- March 2014}}}}
\rhead{}
\lfoot{ISSN: 2231-5373}
\cfoot{\url{http://www.ijmttjournal.org}}
\rfoot{Page \thepage}
}
\newtheorem{theorem}{Theorem}[section]
\newtheorem{lemma}[theorem]{Lemma}

\newtheorem{definition}[theorem]{Definition}
\newtheorem{example}[theorem]{Example}
\newtheorem{remark}[theorem]{Remark}
\title{\textbf{ Cooperating distributed context-free hexagonal array grammar systems with permitting contexts }}
\author {$Sujathakumari K^1$\footnote{email-nksujathakumari@gmail.com} ,\ Dersanambika K.S$^2$\footnote{email-dersanapdf@yahoo.com}\\
\date{}
\small{$^1$ Department of Mathematics}\\\small{S.N College,
Punalur, Kerala 691305, India}\\\small{$^2$ Department of
Mathematics}\\\small{Fatima Mata National College, Kollam 691 001,
India}\\}
\begin{document}
\maketitle
\begin{abstract}{In this paper we associate permitting symbols with rules of Grammars in the components of cooperating
 distributed context-free hexagonal array grammar systems as a control mechanism and investigating the generative power
  of the resulting systems in the terminal mode. This feature of associating permitting symbols with rules when extended
   to patterns in the form of connected arrays also requires checking of symbols, but this is simpler than usual pattern
    matching. The benefit of allowing permitting symbols is that it enables us to reduce the number of components required
     in a cooperating distributed hexagonal array grammar system for generating a set of picture arrays.}
\end{abstract}


\section{Introduction}
In the Context of image analysis and image processing a variety of
generative models for digitalized picture arrays in the two
dimensional plane have been proposed \cite{10}. Out of the
different techniques adopted for various models, grammar based
techniques utilize the rich theory of formal grammars and
languages and develop array grammars generating two dimensional
languages whose elements are picture arrays. There are two
distinct types of array grammars, isometric array grammars and
non-isometric array grammars. Since application of rewriting rule
can increase or decrease the length of the rewritten part, the
dimension of rewritten sub array can change in the case of
non-isometric grammars but application of such a rule is shape
preserving in the case of isometric grammars due to the fact that
the left and right sides of an array rewriting rule is
geometrically identical.

In order to handle more context with rewriting systems, a system
with several components is composed and defined a cooperation
protocol for these components to generate a common sentential
form. Such devices are known as cooperating distributed
(CD)grammar systems \cite{3}. Components are represented by
grammars or other rewriting devices, and the protocol for mutual
cooperation modifying the common sentential form  according to
their own rules. A variety of string grammar system models
\cite{3} have been introduced and studied in the literature.
Rudolf Freund extended the concept of grammar system to arrays
\cite{2} by introducing array grammar system and further J.
Dassow, R. Freund and Gh. P$\breve{a}$un elaborated the power of
cooperation in array grammar system  (cooperating array grammar
system) for various non-context-free sets of arrays which can be
generated in a simple way by cooperating array grammar systems and
simple picture description \cite{1}. They also proved that the
cooperation increases the generative capacity even in the case of
systems with regular array grammar components.

Different kinds of control mechanism that are added to component
grammars for regulated rewriting rules have been considered in
string grammar systems and such control devices are known to
increase the generative power of the grammar in many cases
\cite{1}. Random context grammar is viewed as one of the prototype
mechanism in which components grammars that permit or forbid the
application of a rule based on  the presence or absence of a set
of symbols.

Hexagonal arrays and hexagonal patterns are found in the
literature on picture processing and image analysis. The class of
Hexagonal kolam array language (HKAL) was introduced by Siromoneys
\cite{9}. The class of Hexagonal array language was introduced by
Subramanian. The class of local and recognizable picture languages
were introduced  by Dersanambika et.al. \cite{6}. Recently we
extended cooperative distributed grammar system to Hexagonal
arrays and different capabilities of the system are studied
\cite{8}.

In this paper we associate permitting symbols with rules of the
grammar in the components of cooperating distributed context-free
hexagonal array grammar systems as a control mechanism and
investigating the generative power of the resulting systems in the
terminal mode. This feature of associating permitting symbols with
rules when extended to patterns in the form of connected arrays
also requires checking of symbols, but this is simpler than usual
pattern matching. The benefit of allowing permitting symbols is
that it enables us to reduce the number of components required in
a cooperating distributed hexagonal array grammar system for
generating a set of picture arrays.
\section{Preliminaries and definitions}
Let $V$ be a finite non-empty set of symbols. The set of all
hexagonal arrays made up of elements of $V$ is denoted by
$V^{**H}$. The size of the hexagonal array is defined by the
parameters $LU$(left upper), $LL$(left lower), $RU$(right upper),
$RL$(right lower), $U$(upper), $L$(lower) as shown in Figure 1.
For $X\in V^{**H}$ the length of the left upper side of $X$ is
denoted by $\left| X\right|_{LU}$. Similarly we define
$\left|X\right|_{RU}$,$\left|X\right|_{R}$,$\left|
X\right|_{RL}$,$\left| X\right|_{LL}$ and $\left| X\right|_{L}$.

\begin{figure}
    \begin{center}
        \includegraphics[scale=0.4]{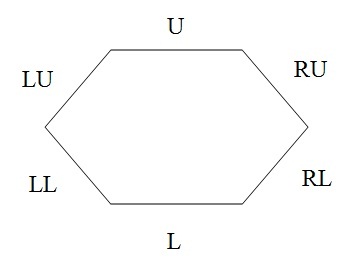}
    \end{center}
    \centerline{Figure 1}
\end{figure}

\begin{definition}An isometric hexagonal array grammar is a construct $G=(N,T,S,P,\#),$ where $N$ and $T$ are disjoint
 alphabets of non terminals and terminals respectively, $S\in N$ is the start symbol, $\#$ is a special symbol called
 blank symbol and $P$ is a finite set of rewriting rules of the form $\alpha\rightarrow\beta$ where $\alpha$ and $\beta$
  are finite subpatterns of a hexagonal pattern over $N\cup T\cup\{\#\}$ satisfying the following conditions:
\begin{enumerate}
\item The shape of $\alpha$ and $\beta$ are identical.
\item $\alpha$ contains at least one element of $N.$ The elements of $T$ appearing in $\alpha$ are not rewritten.
\item A non $\#$ symbol in $\alpha$ is not replaced by a blank symbol in $\beta.$
\item The application of the production $\alpha\rightarrow\beta$ preserves connectivity of the hexagonal array.
\end{enumerate}

For a hexagonal array grammar $G=(N,T,S,P,\#),$ we can define
$x\Longrightarrow y$ for $x,y\in(N\cup T\cup\{\#\})$ if there is a
rule $\alpha\rightarrow\beta\in P$ such that $\alpha$ is a
subpattern of $x$ and $y$ is obtained by replacing $\alpha$ in $x$
by $\beta.$ The reflexive closure of $\Rightarrow$ is denoted by
$\overset{*}{\Rightarrow}.$ The hexagonal array language generated
by $G$ is defined by
$L(G)=\{x\in(T\cup\{\#\})^{**H}:s\overset{*}\Rightarrow x\}.$

\end{definition}
\begin{definition}A hexagonal array grammar is said to be context free if in the rule $\alpha\rightarrow\beta$
\begin{enumerate}
\item non $\#$ symbol in $\alpha$ are not replaced by $\#$ in $\beta$
\item $\alpha$ contain exactly one non-terminal and some occurrences of blank symbol.
\end{enumerate}.The family of languages generated by a context
free hexagonal array grammar is denoted by $CFHA$.
\end{definition}
\vspace{5cm }
\begin{definition}A context free hexagonal array grammar is said to be regular if rules are of the form
$
\begin{array}{*{20}c}
   A & {\#  \to \begin{array}{*{20}c}
   a & B  \\
\end{array}}  \\
\end{array},\begin{array}{*{20}c}
   \#  & {A \to \begin{array}{*{20}c}
   B & a  \\
\end{array}}  \\
\end{array},
$
\[
\begin{array}{*{20}c}
   {\begin{array}{*{20}c}
   {}  \\
   \#   \\
\end{array}} & A  \\
\end{array} \to \begin{array}{*{20}c}
   {\begin{array}{*{20}c}
   {}  \\
   B  \\
\end{array}} & a  \\
\end{array},\begin{array}{*{20}c}
   A & {\begin{array}{*{20}c}
   \#   \\
   {}  \\
\end{array}}  \\
\end{array} \to \begin{array}{*{20}c}
   a & {\begin{array}{*{20}c}
   B  \\
   {}  \\
\end{array}}  \\
\end{array},\begin{array}{*{20}c}
   A & {\begin{array}{*{20}c}
   {}  \\
   \#   \\
\end{array}}  \\
\end{array} \to \begin{array}{*{20}c}
   a & {\begin{array}{*{20}c}
   {}  \\
   B  \\
\end{array},}  \\
\end{array}
\]
$
\begin{array}{*{20}c}
   {\begin{array}{*{20}c}
   \#   \\
   {}  \\
\end{array}} & A  \\
\end{array} \to \begin{array}{*{20}c}
   {\begin{array}{*{20}c}
   B  \\
   {}  \\
\end{array}} & {a,A \to a}  \\
\end{array}
$

The family of languages generated by a regular hexagonal array grammar is denoted by $REGHA.$
\end{definition}
\begin{definition}A cooperating hexagonal array grammar system (of type $X,X\in\{CFHA,REGHA\},$ and degree $n,n\ge 1$), is a construct $\Gamma=(N,T,S,P_1,P_2,\ldots,P_n),$ where $N$ and $T$ are non-terminal and terminal alphabets respectively,
$S \in N$ and $P_1,P_2,\ldots,P_n$ are finite sets of regular
respectively context free rules over $N\cup T.$
\end{definition}
\begin{definition}Let $\Gamma$ be a cooperating hexagonal array grammar system. Let $x,y\in \Gamma^*.$ Then we write $x\underset{p_i}{\overset{k}{\Rightarrow}}y$ if and only if there are words $x_1,x_2,\ldots,x_{k+1}$ such that
\begin{enumerate}
\item $x=x_1,y=x_{k+1},$ \item $x_j\underset{p_i}\Rightarrow
x_{j+1},$ that is,
$x_j=x_j'A_jx_j'',\,x_{j+1}=x_j'W_jx_j'',\,A_j\rightarrow W_j\in
P_i,\,$ $1\le j\le k.$
\end{enumerate}

Moreover, we write\\
\indent\indent$x\underset{P_i}{\overset{\le k}{\Rightarrow}}y$ if and only if $x\underset{P_i}{\overset{k'}{\Rightarrow}}y$ for some $k'\le k,$\\
\indent\indent $x\underset{P_i}{\overset{\ge k}{\Rightarrow}}y$ if and only if $x\underset{P_i}{\overset{\le k'}{\Rightarrow}}y$ for some $k'\ge k,$\\
\indent\indent $x\underset{P_i}{\overset{*}{\Rightarrow}}y$ if and only if $x\underset{P_i}{\overset{k}{\Rightarrow}}y$ for some $k$\\
\indent\indent $x\underset{P_i}{\overset{t}{\Rightarrow}}y$ if and only if $x\underset{P_i}{\overset{*}{\Rightarrow}}y$ and there is no $z\ne y$ with $y\underset{P_i}{\overset{*}{\Rightarrow}}z.$

By $CD_n(X,f)$ we denote family of hexagonal array language
generated by cooperating hexagonal array grammar system consisting
of at most $n$ components of type $X\in(REGHA,CFHA)$ in the mode
$f.$
\end{definition}
\begin{definition}
A random context grammar is a quadruple $G=(N,T,P,S)$ where $N$ is
the alphabet of non-terminals, $T$ is the alphabet of terminals
such that $N\cap T=\phi, V=N\cup T,S\in N$ is the start symbol,and
$P$ is a finite set of productions of the form $(A\rightarrow
x,Per,For)$ where $A\rightarrow x$ is a context free
production,$A\in N$ and $x\in V^{\ast}$,and $Per,For\subseteq N$.
For $U,V\in V^{\ast}$ and a production $(A\rightarrow
x,Per,For)\in P$, the relation $Av\Rightarrow uxv$ holds provided
that $per\subseteq alph(uv)$ and $alph(uv)\cap For=\phi$. A
permitting (forbidding)grammar is a random context grammar
$G=(N,T,P,S)$ where for each production\\ $(A\rightarrow
x,Per,For)\in P$, it holds that $For=\phi ~(Per=\phi~
respectively)$
\end{definition}

\section {Cooperating distributed context-free hexagonal array grammar system with permitting symbols}

The set of all symbols in the labeled cells of the array $p$ is denoted by alph$(p)$ A permitting CF hexagonal
 array rule is an array grammar $G$ is of the form $(\alpha\rightarrow\beta,per),$ where $\alpha\rightarrow\beta$
 is a context-free hexagonal array rewriting rule and $per\subseteq N$, where $N$is the set of non -terminals of the grammar.
 If $per=\phi$, then we avoid mentioning it in the rule.For any two arrays $p,q$ and a  permitting CF hexagonal array rule
 $(\alpha\rightarrow\beta,per)$, the array $q$ is derived from $p$ by replacing $\alpha$ in $p$ by $\beta$ provided that
 $per \subseteq alph(p\backslash\alpha)$. A permitting cooperating distributed context-free hexagonal array grammar systems
 (pCDCFHAGS) is $\Gamma=(N,T,P_1, P_2 ,\ldots,P_N,S),\,n\geq 1$ where $N$ is a finite  set of non terminals,$S\in N$ is
 the start symbol,$T$ is a finite set of terminals, $N\bigcap T=\phi$ and each $P_i$, for $1\leq i \leq n $ is a finite set
 of permitting CF hexagonal array rewriting rules.

For any two hexagonal arrays $p,q$, we denote $p\underset{p_i}{\overset{t}{\Rightarrow}}q$
an array rewriting step performed by applying a permitting cooperating CF Hexagonal array rule in $p_i,1\leq i \leq n$,
 and by $p\underset{p_i}{\overset{*}{\Rightarrow}}q$  the transitive closure of $p\underset{p_i}{\overset{t}{\Rightarrow}}q.$
 Also we say that the array $p$ derives an array $q$ in the terminal mode or $t$ mode and write $p\underset{p_i}{\overset{t}{\Rightarrow}}q$, if$p\underset{p_i}{\overset{*}{\Rightarrow}}q,$ and there is no array $s$ such that $q\underset{p_i}{\overset{t}{\Rightarrow}}s$ The array language generated by $\Gamma$ in the $t$ mode is defined as
$L(\Gamma)=\{q\colon s=
p_o\underset{p_{i_1}}{\overset{t}{\Rightarrow}}
p_1\underset{p_{i_2}}{\overset{t}{\Rightarrow}}
p_2,\ldots,\underset{p_{i_m}}{\overset{t}{\Rightarrow}} p_m =q\in
V^{**H},\,m \geq 1,\,i_j\in \{1,2,\ldots,n\}$ for $1\le j\le m\}.$

Note that $i_1,i_2,\ldots,i_m$ is any sequence of symbols belonging to $\{1,2,\ldots,n\}$ where repeated symbols are allowed.
 Also $pCD_n(HCFA,t)$ denote the family of array languages generated in the $t$ mode by permitting cooperating CF hexagonal array grammar systems with at most $n$ components
\begin{example}Consider the context-free hexagonal array grammars with rules
$G=(\{S,A,B\},\{a\},P,S,\#)$ where
\[
\begin{array}{l}
 P = \left\{ {\begin{array}{*{20}c}
   {1)\;S} & {\begin{array}{*{20}c}
   \#   \\
   \#   \\
\end{array}}  \\
\end{array}} \right. \Rightarrow \begin{array}{*{20}c}
   a & {\begin{array}{*{20}c}
   A  \\
   B  \\
\end{array}}  \\
\end{array},\begin{array}{*{20}c}
   {2)\;A} & {\begin{array}{*{20}c}
   \#   \\
   {}  \\
\end{array}}  \\
\end{array} \to \begin{array}{*{20}c}
   a & {\begin{array}{*{20}c}
   {A'}  \\
   {}  \\
\end{array}}  \\
\end{array}, \\
 \left. {\begin{array}{*{20}c}
   {3)\;B} & {\begin{array}{*{20}c}
   {}  \\
   \#   \\
\end{array}}  \\
\end{array} \to \begin{array}{*{20}c}
   a & {\begin{array}{*{20}c}
   {}  \\
   {B'}  \\
\end{array}}  \\
\end{array},\;4)\;A' \to A,\;5)\;B' \to B,\;6)\;A \to a,\;7)\;B \to a} \right\} \\
 \end{array}
\]
\[
\begin{array}{*{20}c} \\
   {} & {} & {} & {} & {} & {} & {} & {}  \\
   {} & {} & {} & {} & {} & a & {} & {}  \\
   {} & {} & {} & {} & a & {} & {} & {}  \\
   {} & {} & {} & . & {} & {} & {} & {}  \\
   {} & {} & . & {} & {} & {} & {} & {}  \\
   {} & a & {} & {} & {} & {} & {} & {}  \\
   a & {} & {} & {} & {} & {} & {} & {}  \\
   {} & a & {} & {} & {} & {} & {} & {}  \\
   {} & {} & a & {} & {} & {} & {} & {}  \\
   {} & {} & {} & {} & {} & {} & {} & {}  \\
\end{array}
\]
\centerline{Figure 2}
 \par $G$ generates hexagonal arrays over $\{a\}$
is in the shape of left arrow head  but size of left upper arm and
left lower arm are not necessarily equal.

\end{example}
\begin{example}
The pCDCFHAGS
$G_{1}=(\{S,A,B,A',B',C,D\},\{a\},P_{1},S)$ where $P_{1}$ consists  of the following rules.
\[
\begin{array}{l}
 1.\begin{array}{*{20}c}
   S & {\begin{array}{*{20}c}
   \#   \\
   \#   \\
\end{array}}  \\
\end{array} \Rightarrow \begin{array}{*{20}c}
   a & {\begin{array}{*{20}c}
   A  \\
   B  \\
\end{array}}  \\
\end{array}, \\\\
 2.\left( {\begin{array}{*{20}c}
   A & {\begin{array}{*{20}c}
   \#   \\
   {}  \\
\end{array}}  \\
\end{array} \to \begin{array}{*{20}c}
   a & {\begin{array}{*{20}c}
   {A'}  \\
   {}  \\
\end{array},}  \\
\end{array}\left\{ B \right\}} \right), \\\\
 3.\left( {\begin{array}{*{20}c}
   B & {\begin{array}{*{20}c}
   {}  \\
   \#   \\
\end{array} \to \begin{array}{*{20}c}
   a & {\begin{array}{*{20}c}
   {}  \\
   {B'}  \\
\end{array},\left\{ {A'} \right\}}  \\
\end{array}}  \\
\end{array}} \right), \\\\
 4.\left( {A' \to A,\left\{ {B'} \right\}} \right), \\\\
 5.\left( {B' \to B,\left\{ A \right\}} \right), \\\\
 6.\left( {A \to C,\left\{ B \right\}} \right), \\\\
 7.\left( {B \to D,\left\{ C \right\}} \right), \\\\
 8.C \to a, \\\\
 9.D \to a \\
 \end{array}
\]
generates (in the $t$ mode) the set $H_{LA}$ of all arrays over
$\{a\}$ in the shape of a left arrow head with left upper arm and
left lower arm are equal in size (Figure 3).
\[
\begin{array}{*{20}c}
   {} & {} & {} & a & {} & {} & {}  \\
   {} & {} & a & {} & {} & {} & {}  \\
   {} & a & {} & {} & {} & {} & {}  \\
   a & {} & {} & {} & {} & {} & {}  \\
   {} & a & {} & {} & {} & {} & {}  \\
   {} & {} & a & {} & {} & {} & {}  \\
   {} & {} & {} & a & {} & {} & {}  \\
\end{array}
\]
\centerline{Figure 3} The derivations starts with rule 1 followed
by rule 2 which can be applied as the permitting symbol $B$ is
present in the array. This grows left upper arm (LU) by one cell.
Then rule 3 can be applied due to the presence of permitting
symbol $A'$ and this grows left lower arm (LL) by one cell. An
application of rule 4 followed by 5, again noting that the
permitting symbols of the respective rules are present changes
$A'$ to $A$ and $B'$ to $B$. the repeated application of the
process growing both the left upper arm and left lower arm equal
in size. Rules 6 and 7 are applied changing $A$ to $C$ and $B$ to
$D$ so that the derivation can be terminated by the application of
rules 8 and 9 thus yielding a hexagonal array in the shape of a
left arrow head with size of $LU$ and $LL$ are equal.
\end{example}

\begin{remark}A $pCDCFHAGS$ where the set of permitting symbols in all the components is empty, is simply a cooperating distributed $CF$ hexagonal array system ($CDCFHAGS$)\cite{8}. The family of array languages generated in the $t$ mode by a $CDCFHAGS$ with at most $n$ components is denoted by $CD_n(CFHA,t).$ If the rules in all the components are only in the form of rules of a regular array grammar, then this family is denoted by $CD_n(REG
HA,t).$
\end{remark}
We now show that the set $H_{HF}$ of all $n\times n\times n(n\ge
3)$ arrays over $\{a\}$ in the form of hollow hexagonal frame.
Figure 3 can be generated (in the $t$ mode) by a $pCDCFHAGS$ with
only two components.

\begin{lemma} $H_{HF}\in pCD_2(HCFA,t).$
\end{lemma}
\begin{proof}The set $H_{HF}$ is generated (in the $t$ mode) by the
$pCDCFHAGS\,\,G_2=(\{S,A,B,A',B',C,D,C',D',E,F,E',F',X,Y\},\{a\},P_1,P_2,S)$.
 The rules in the component$P_1$ are given by
\[
\begin{array}{l}
 1.\begin{array}{*{20}c}
   S & {\begin{array}{*{20}c}
   \#   \\
   \#   \\
\end{array}}  \\
\end{array} \Rightarrow \begin{array}{*{20}c}
   a & {\begin{array}{*{20}c}
   A  \\
   B  \\
\end{array}}  \\
\end{array}, \\
 2.\left( {\begin{array}{*{20}c}
   A & {\begin{array}{*{20}c}
   \#   \\
   {}  \\
\end{array}}  \\
\end{array} \to \begin{array}{*{20}c}
   a & {\begin{array}{*{20}c}
   {A'}  \\
   {}  \\
\end{array},}  \\
\end{array}\left\{ B \right\}} \right), \\
 3.\left( {\begin{array}{*{20}c}
   B & {\begin{array}{*{20}c}
   {}  \\
   \#   \\
\end{array} \to \begin{array}{*{20}c}
   a & {\begin{array}{*{20}c}
   {}  \\
   {B'}  \\
\end{array},\left\{ {A'} \right\}}  \\
\end{array}}  \\
\end{array}} \right), \\\\
 4.\left( {A' \to A,\left\{ {B'} \right\}} \right), \\\\
 5.\left( {B' \to B,\left\{ A \right\}} \right), \\\\
 6.\left( {\begin{array}{*{20}c}
   {A'} & {\#  \to \begin{array}{*{20}c}
   a & {C,\left\{ {B'} \right\}}  \\
\end{array}}  \\
\end{array}} \right), \\\\
 7.\left( {\begin{array}{*{20}c}
   {B'} & {\#  \to \begin{array}{*{20}c}
   a & {D,\left\{ C \right\}}  \\
\end{array}}  \\
\end{array}} \right), \\\\
 8.\left( {\begin{array}{*{20}c}
   C & {\#  \to \begin{array}{*{20}c}
   a & {C',\left\{ D \right\}}  \\
\end{array}}  \\
\end{array}} \right), \\\\
 9.\left( {\begin{array}{*{20}c}
   D & {\#  \to \begin{array}{*{20}c}
   a & {D',\left\{ {C'} \right\}}  \\
\end{array}}  \\
\end{array}} \right), \\\\
 10.\left( {C' \to C,\left\{ {D'} \right\}} \right), \\\\
 11.\left( {D' \to D,\left\{ C \right\}} \right), \\\\
 12.\left( {\begin{array}{*{20}c}
   {C'} & {\begin{array}{*{20}c}
   {}  \\
   \#   \\
\end{array} \to \begin{array}{*{20}c}
   a & {\begin{array}{*{20}c}
   {}  \\
   E  \\
\end{array},\left\{ {D'} \right\}}  \\
\end{array}}  \\
\end{array}} \right), \\\\
13.\left( {\begin{array}{*{20}c}
   {D'} & {\begin{array}{*{20}c}
   \#   \\
   {}  \\
\end{array} \to \begin{array}{*{20}c}
   a & {\begin{array}{*{20}c}
   F  \\
   {}  \\
\end{array},\left\{ E \right\}}  \\
\end{array}}  \\
\end{array}} \right), \\\\
 14.\left( {\begin{array}{*{20}c}
   E & {\begin{array}{*{20}c}
   {}  \\
   \#   \\
\end{array} \to \begin{array}{*{20}c}
   a & {\begin{array}{*{20}c}
   {}  \\
   {E'}  \\
\end{array},\left\{ F \right\}}  \\
\end{array}}  \\
\end{array}} \right), \\\\
 15.\left( {\begin{array}{*{20}c}
   F & {\begin{array}{*{20}c}
   \#   \\
   {}  \\
\end{array} \to \begin{array}{*{20}c}
   a & {\begin{array}{*{20}c}
   {F'}  \\
   {}  \\
\end{array},\left\{ {E'} \right\}}  \\
\end{array}}  \\
\end{array}} \right), \\\\
16.\left( {E' \to E,\left\{ {F'} \right\}} \right), \\
 \end{array}
 \]
 \[
 \begin{array}{l}
 17.\left( {F' \to F,\left\{ E \right\}} \right), \\\\
 18.\left( {\begin{array}{*{20}c}
   E & {\begin{array}{*{20}c}
   {}  \\
   \#   \\
\end{array} \to \begin{array}{*{20}c}
   a & {\begin{array}{*{20}c}
   {}  \\
   X  \\
\end{array},\left\{ F \right\}}  \\
\end{array}}  \\
\end{array}} \right), \\\\
 19.\left( {\begin{array}{*{20}c}
   F & {\begin{array}{*{20}c}
   \#   \\
   {}  \\
\end{array} \to \begin{array}{*{20}c}
   a & {\begin{array}{*{20}c}
   Y  \\
   {}  \\
\end{array},\left\{ X \right\}}  \\
\end{array}}  \\
\end{array}} \right), \\
 \end{array}
\]The rules in the component $P_2$ are given by
\[
P_2  = \left\{ {X \to a,Y \to a,F \to a} \right\}
\]
\[
\begin{array}{*{20}c}
   {} & {} & {} & {} & a & {} & a & {} & a & {} & a & {} & a & {} & {} & {} & {}  \\
   {} & {} & {} & a & {} & {} & {} & {} & {} & {} & {} & {} & {} & a & {} & {} & {}  \\
   {} & {} & a & {} & {} & {} & {} & {} & {} & {} & {} & {} & {} & {} & a & {} & {}  \\
   {} & a & {} & {} & {} & {} & {} & {} & {} & {} & {} & {} & {} & {} & {} & a & {}  \\
   a & {} & {} & {} & {} & {} & {} & {} & {} & {} & {} & {} & {} & {} & {} & {} & a  \\
   {} & a & {} & {} & {} & {} & {} & {} & {} & {} & {} & {} & {} & {} & {} & a & {}  \\
   {} & {} & a & {} & {} & {} & {} & {} & {} & {} & {} & {} & {} & {} & a & {} & {}  \\
   {} & {} & {} & a & {} & {} & {} & {} & {} & {} & {} & {} & {} & a & {} & {} & {}  \\
   {} & {} & {} & {} & a & {} & a & {} & a & {} & a & {} & a & {} & {} & {} & {}  \\
\end{array}
\]
\centerline{Figure 4} Using $t$- mode of derivation, starting with
the symbol $S$ an application of rule (1) in the first component
followed by rule(2) which can be applied as the permitting symbol
$B$ is present in the array, grows in the $LU$ arm by one place.
Since $A'$ is the permuting symbol for rule (3), rule (3) can then
applied which results the growth of $LL$ arm by one place. Now the
situations are ready for applying rules (4) and (5) and at this
stage $A'$ becomes $A$ and $B'$ becomes $B$. The process can be
repeated and this in turn results the growth of $LU$ and $LL$ arms
equal in length. Instead of rule (4) rule (6) is applied followed
by (7),(8),(9),(10),(11) allows upper and lower arms to grow in
equal length, with permitting symbols in all these rules directing
the sequence of applications in the right order. If rule (12) is
used instead of rule (10) and this is followed by rule (13) Right
upper ($RU$) and Right lower ($RL$) arms grows equal in size and
correct application of rule (18) and (19) will result in the
symbol $X$ in the $RU$ arm and $F$ in the lower right arm  where
$X$ is at the position of right end point of $RU$ and $RL$ arm so
that further application of productions in $P_{1}$ is not possible
at any non-terminals. At this stage applying productions in
$P_{2}$ and this in turn terminate the derivation yielding a
hollow hexagon with its parallel arms are equal in size.
\end{proof}

\begin{lemma}
\begin{enumerate}
\item $H_{HF}\in CD_{3}(CFHA,t)$ \item $H_{HF}\notin
CD_{n}(REHA,t)~\mbox{for}~n\geq 1$
\end{enumerate}
\end{lemma}
Proof follows from the result in \cite{8}.
\\\\\\
\begin{theorem}
\begin{enumerate}
\item $CFHA= CD_{1}(CFHA,t) \subset pCD_{1}(CFHA,t).$ \item
$CD_{2}(CFHA,t) \subset pCD_{2}(CFHA,t).$ \item $
pCD_{2}(CFHA,2)\setminus CD_{n}(REGHA,t) \neq\phi $ for any $n\geq
1.$
\end{enumerate}
\begin{proof}
The equality follows from the results from \cite{8}. We know that,
a $pCDCFHAGS$ with empty set of permitting symbols associated with
the rules is same as a cooperating distributed $CF$ hexagonal
array grammar system and so $CD_{1}(CFHA,t) =pCD_{1}(CFHA,t)$ if
$per=\phi$. Examples (1) and (2) illustrated the fact that the set
$H_{LA}$ of all hexagonal arrays over ${a}$ in the shape of left
arrow head with left upper arm and left lower arm with equal size
is generated by the $pCD_{1}CFHAGS$ with only one component and
working in the $t$-mode and hence the inclusion is proper which
proves (1). Similar arguments for inclusion in statement(2) are
hold. From the proof of the lemma(1) it is very clear that under
the strict application of the derivation rules with the respective
permitive symbols generate a hollow hexagon with parallel arms
equal in size. Such a generation is not possible in $
CD_{2}CFHAGS$ since $ per =\phi$  and incorrect application of
terminating rule will leads to non-completion of the hollow
hexagon with parallel arms equal in size. Thus $CD_{2}(CFHA,t)
\subset pCD_{2}(CFHA,t).$

Consider the array languages generated by $pCDCFHAGS$ in example
(1) and in the proof of lemma (1). It can be seen that generating
the arrow head of patterns of the language, we should require two
growing heads at the same time. But in the $ CDREGHAGS$ with any
number of components the array rules contains only one growing
head. So the same language cannot be generated by
$CD_{n}(REGHA,t)$ and this proves (3).
\end{proof}
\end{theorem}
To show the power of the cooperating hexagonal array grammar system with the array rules controlled by permitting symbols, consider the following example of a language of set of all hexagons with parallel arms are equal in length over a one letter alphabet.
\begin{example}Consider the $pCDCFHAGS$\\
$G_4=(\{S,S',A,B,A',B',C,D,C',D',E',F',G,H,I,I',J,J',K,K',L,L',\\~~~~~~~~~~~M,M',N,N'O,R,T,T',X\},\{a,b\},p_1,p_2,S)$.\\
 The rules in the component $p_1$ are\\
\[
\begin{array}{l}
 1)\begin{array}{*{20}c}
   S & {\begin{array}{*{20}c}
   \#   \\
   \#   \\
\end{array}} & \#   \\
\end{array} \to \begin{array}{*{20}c}
   a & {\begin{array}{*{20}c}
   A  \\
   B  \\
\end{array}} & {S'}  \\
\end{array}, \\\\
 2)\left( {\begin{array}{*{20}c}
   A & {\begin{array}{*{20}c}
   \#   \\
   {}  \\
\end{array} \to \begin{array}{*{20}c}
   a & {\begin{array}{*{20}c}
   {A'}  \\
   {}  \\
\end{array},\left\{ B \right\}}  \\
\end{array}}  \\
\end{array}} \right), \\
\end{array}
\]
\[
\begin{array}{l}
 3)\left( {\begin{array}{*{20}c}
   B & {\begin{array}{*{20}c}
   {}  \\
   \#   \\
\end{array} \to \begin{array}{*{20}c}
   a & {\begin{array}{*{20}c}
   {}  \\
   {B'}  \\
\end{array},\left\{ {A'} \right\}}  \\
\end{array}}  \\
\end{array}} \right), \\\\
 4)\left( {A' \to A,\left\{ {B'} \right\}} \right), \\\\
 5)\left( {B' \to B,\left\{ A \right\}} \right), \\\\
 6)\left( {\begin{array}{*{20}c}
   {A'} & \#   \\
\end{array} \to \begin{array}{*{20}c}
   a & C  \\
\end{array},\left\{ {B'} \right\}} \right), \\\\
 7)\left( {\begin{array}{*{20}c}
   {B'} & \#   \\
\end{array} \to \begin{array}{*{20}c}
   a & D  \\
\end{array},\left\{ C \right\}} \right), \\\\
 8)\left( {\begin{array}{*{20}c}
   C & \#   \\
\end{array} \to \begin{array}{*{20}c}
   a & {C'}  \\
\end{array},\left\{ D \right\}} \right), \\\\
 9)\left( {\begin{array}{*{20}c}
   D & \#   \\
\end{array} \to \begin{array}{*{20}c}
   a & {D'}  \\
\end{array},\left\{ {C'} \right\}} \right), \\\\
 10)\left( {C' \to C,\left\{ {D'} \right\}} \right), \\\\
 11)\left( {D' \to D,\left\{ C \right\}} \right), \\\\
 12)\left( {\begin{array}{*{20}c}
   {\begin{array}{*{20}c}
   {}  \\
   \#   \\
\end{array}} & C & \#   \\
\end{array} \to \begin{array}{*{20}c}
   {\begin{array}{*{20}c}
   {}  \\
   {E'}  \\
\end{array}} & a & a  \\
\end{array},\left\{ D \right\}} \right), \\\\
13)\left( {\begin{array}{*{20}c}
   {\begin{array}{*{20}c}
   \#   \\
   {}  \\
\end{array}} & D & \#   \\
\end{array} \to \begin{array}{*{20}c}
   {\begin{array}{*{20}c}
   {F'}  \\
   {}  \\
\end{array}} & a & a  \\
\end{array},\left\{ {E'} \right\}} \right), \\\\
14)\left( {E' \to E,\left\{ {F'} \right\}} \right), \\\\
 15)\left( {F' \to F,\left\{ E \right\}} \right), \\\\
 16)\left( {\begin{array}{*{20}c}
   {\begin{array}{*{20}c}
   {}  \\
   \#   \\
\end{array}} & E  \\
\end{array} \to \begin{array}{*{20}c}
   {\begin{array}{*{20}c}
   {}  \\
   {E'}  \\
\end{array}} & a  \\
\end{array},\left\{ F \right\}} \right), \\\\
 17)\left( {\begin{array}{*{20}c}
   {\begin{array}{*{20}c}
   \#   \\
   {}  \\
\end{array}} & F  \\
\end{array} \to \begin{array}{*{20}c}
   {\begin{array}{*{20}c}
   {F'}  \\
   {}  \\
\end{array}} & a  \\
\end{array},\left\{ {E'} \right\}} \right), \\\\
 18)\left( {\begin{array}{*{20}c}
   {E'} & {\begin{array}{*{20}c}
   {}  \\
   F  \\
\end{array}} & H  \\
\end{array} \to \begin{array}{*{20}c}
   a & {\begin{array}{*{20}c}
   {}  \\
   a  \\
\end{array}} & T  \\
\end{array},\left\{ {S'} \right\}} \right), \\\\
19)\left( {\begin{array}{*{20}c}
   {S'} & {\begin{array}{*{20}c}
   \#   \\
   \#   \\
\end{array}} & \#   \\
\end{array} \to \begin{array}{*{20}c}
   a & {\begin{array}{*{20}c}
   G  \\
   H  \\
\end{array}} & {S'}  \\
\end{array},\left\{ T \right\}} \right), \\
\end{array}
 \]
 \[
\begin{array}{l}
 20)\left( {\begin{array}{*{20}c}
   T & {\begin{array}{*{20}c}
   \#   \\
   \#   \\
\end{array}}  \\
\end{array} \to \begin{array}{*{20}c}
   a & {\begin{array}{*{20}c}
   I  \\
   J  \\
\end{array}} & {T'}  \\
\end{array},\left\{ {S'} \right\}} \right), \\\\
 21)\begin{array}{*{20}c}
   T & {\begin{array}{*{20}c}
   \#   \\
   \#   \\
\end{array}}  \\
\end{array} \to \begin{array}{*{20}c}
   a & {\begin{array}{*{20}c}
   I  \\
   J  \\
\end{array}}  \\
\end{array}, \\\\
 22)\left( {\begin{array}{*{20}c}
   G & {\begin{array}{*{20}c}
   \#   \\
   {}  \\
\end{array}}  \\
\end{array} \to \begin{array}{*{20}c}
   a & {\begin{array}{*{20}c}
   G  \\
   {}  \\
\end{array}}  \\
\end{array},\left\{ H \right\}} \right), \\
 23)\left( {\begin{array}{*{20}c}
   H & {\begin{array}{*{20}c}
   {}  \\
   \#   \\
\end{array}}  \\
\end{array} \to \begin{array}{*{20}c}
   a & {\begin{array}{*{20}c}
   {}  \\
   H  \\
\end{array}}  \\
\end{array},\left\{ G \right\}} \right), \\\\
 24)\left( {\begin{array}{*{20}c}
   I & {\begin{array}{*{20}c}
   \#   \\
   {}  \\
\end{array}}  \\
\end{array} \to \begin{array}{*{20}c}
   a & {\begin{array}{*{20}c}
   {I'}  \\
   {}  \\
\end{array}}  \\
\end{array},\left\{ J \right\}} \right), \\
 25)\left( {\begin{array}{*{20}c}
   J & {\begin{array}{*{20}c}
   {}  \\
   \#   \\
\end{array}}  \\
\end{array} \to \begin{array}{*{20}c}
   a & {\begin{array}{*{20}c}
   {}  \\
   {J'}  \\
\end{array}}  \\
\end{array},\left\{ {I'} \right\}} \right), \\\\
 26)\left( {I' \to I,\left\{ {J'} \right\}} \right), \\\\
 27)\left( {J' \to J,\left\{ I \right\}} \right), \\\\
 28)\begin{array}{*{20}c}
   I & {\begin{array}{*{20}c}
   a  \\
   {}  \\
\end{array}} & \#   \\
\end{array} \to \begin{array}{*{20}c}
   a & {\begin{array}{*{20}c}
   a  \\
   {}  \\
\end{array}} & K  \\
\end{array}, \\
 29)\begin{array}{*{20}c}
   J & {\begin{array}{*{20}c}
   {}  \\
   a  \\
\end{array}} & \#   \\
\end{array} \to \begin{array}{*{20}c}
   a & {\begin{array}{*{20}c}
   {}  \\
   a  \\
\end{array}} & L  \\
\end{array}, \\\\
 30)\left( {\begin{array}{*{20}c}
   K & \#   \\
\end{array} \to \begin{array}{*{20}c}
   a & {K'}  \\
\end{array},\left\{ L \right\}} \right), \\\\
 31)\left( {\begin{array}{*{20}c}
   L & \#   \\
\end{array} \to \begin{array}{*{20}c}
   a & {L'}  \\
\end{array},\left\{ {K'} \right\}} \right), \\\\
 32)\left( {K' \to K,\left\{ {L'} \right\}} \right), \\\\
 33)\left( {\begin{array}{*{20}c}
   {L'} & \#   \\
\end{array} \to L,\left\{ K \right\}} \right), \\\\
 34)\left( {\begin{array}{*{20}c}
   K & {\begin{array}{*{20}c}
   {}  \\
   \#   \\
\end{array}}  \\
\end{array} \to \begin{array}{*{20}c}
   a & {\begin{array}{*{20}c}
   {}  \\
   {M'}  \\
\end{array}}  \\
\end{array},\left\{ L \right\}} \right), \\\\
 35)\left( {\begin{array}{*{20}c}
   L & {\begin{array}{*{20}c}
   \#   \\
   {}  \\
\end{array}}  \\
\end{array} \to \begin{array}{*{20}c}
   a & {\begin{array}{*{20}c}
   {N'}  \\
   {}  \\
\end{array}}  \\
\end{array},\left\{ {M'} \right\}} \right), \\\\
 36)\left( {M' \to M,\left\{ {N'} \right\}} \right), \\
\end{array}
\]
\[
\begin{array}{l}
 37)\left( {N' \to N,\left\{ M \right\}} \right), \\\\
 38)\left( {\begin{array}{*{20}c}
   M & {\begin{array}{*{20}c}
   {}  \\
   \#   \\
\end{array}}  \\
\end{array} \to \begin{array}{*{20}c}
   a & {\begin{array}{*{20}c}
   {}  \\
   {M'}  \\
\end{array}}  \\
\end{array},\left\{ N \right\}} \right), \\\\
 39)\left( {\begin{array}{*{20}c}
   N & {\begin{array}{*{20}c}
   \#   \\
   {}  \\
\end{array}}  \\
\end{array} \to \begin{array}{*{20}c}
   a & {\begin{array}{*{20}c}
   {N'}  \\
   {}  \\
\end{array}}  \\
\end{array},\left\{ {M'} \right\}} \right), \\\\
 40)\left( {\begin{array}{*{20}c}
   N & {\begin{array}{*{20}c}
   \#   \\
   {}  \\
\end{array}}  \\
\end{array} \to \begin{array}{*{20}c}
   a & {\begin{array}{*{20}c}
   {N'}  \\
   {}  \\
\end{array}}  \\
\end{array},\left\{ {M'} \right\}} \right), \\\\
 41)\left( {\begin{array}{*{20}c}
   N & {\begin{array}{*{20}c}
   {M'}  \\
   {}  \\
\end{array}}  \\
\end{array} \to \begin{array}{*{20}c}
   a & {\begin{array}{*{20}c}
   {N'}  \\
   {}  \\
\end{array}}  \\
\end{array},\left\{ {M'} \right\}} \right), \\\\
 42)\left( {\begin{array}{*{20}c}
   T & {\begin{array}{*{20}c}
   \#   \\
   \#   \\
\end{array}} & \#   \\
\end{array} \to \begin{array}{*{20}c}
   a & {\begin{array}{*{20}c}
   O  \\
   R  \\
\end{array}} & a  \\
\end{array},\left\{ X \right\}} \right), \\\\
 43)\left( {\begin{array}{*{20}c}
   R & {\begin{array}{*{20}c}
   {}  \\
   \#   \\
\end{array}}  \\
\end{array} \to \begin{array}{*{20}c}
   a & {\begin{array}{*{20}c}
   {}  \\
   R  \\
\end{array}}  \\
\end{array},\left\{ O \right\}} \right), \\\\
 44)\begin{array}{*{20}c}
   T & {\begin{array}{*{20}c}
   a  \\
   a  \\
\end{array}} & \#   \\
\end{array} \to \begin{array}{*{20}c}
   a & {\begin{array}{*{20}c}
   a  \\
   a  \\
\end{array}} & T  \\
\end{array}, \\\\
 45)\begin{array}{*{20}c}
   {S'} & {\begin{array}{*{20}c}
   E  \\
   F  \\
\end{array}} & \#   \\
\end{array} \to \begin{array}{*{20}c}
   a & {\begin{array}{*{20}c}
   a  \\
   a  \\
\end{array}} & T  \\
\end{array} \\
 \end{array}
\]
The rules in the component $p_2$ are given by
\[
P_2  = \left\{ {G \to a,H \to a,O \to a,T \to a,X \to a,R \to a}
\right\}
\]

\[
\begin{array}{*{20}c}
   {} & {} & {} & {} & a & {} & a & {} & a & {} & a & {} & a & {} & {} & {} & {}  \\
   {} & {} & {} & a & {} & a & {} & a & {} & a & {} & a & {} & a & {} & {} & {}  \\
   {} & {} & a & {} & a & {} & a & {} & a & {} & a & {} & a & {} & a & {} & {}  \\
   {} & a & {} & a & {} & a & {} & a & {} & a & {} & a & {} & a & {} & a & {}  \\
   a & {} & a & {} & a & {} & a & {} & a & {} & a & {} & a & {} & a & {} & a  \\
   {} & a & {} & a & {} & a & {} & a & {} & a & {} & a & {} & a & {} & a & {}  \\
   {} & {} & a & {} & a & {} & a & {} & a & {} & a & {} & a & {} & a & {} & {}  \\
   {} & {} & {} & a & {} & a & {} & a & {} & a & {} & a & {} & a & {} & {} & {}  \\
   {} & {} & {} & {} & a & {} & a & {} & a & {} & a & {} & a & {} & {} & {} & {}  \\
\end{array}
\]
\centerline{Figure 5} Starting with $S,$ repeated application of
the first five rules of the component $P_1$ in this order generate
$<a$ having equal arms, with $LU$ arm having non-terminal $A$ and
$LL$ arm having non-terminal $B.$ Once two rules $(A'\#\rightarrow
aC,\{B'\})$ and $(B'\#\rightarrow aD,\{C\})$ of component $P_1$
are used, the generations of these two arms end with terminal $a$
and then starts the generations of upper and lower arms of the
hexagon using rules (6) to (11). Then the right application of
rule (12) to (17) then (18), (19), $\dots,$ (45) subjected to the
permitting symbols will result to a hexagonal picture and finally
by the application of rules in $P_2,$ we get the required hexagon
over the one letter alphabet $\{a\}$ as in Figure 5.
\end{example}
 \par In the Siromoney matrix grammar (9) rectangular arrays are generated in
two phases; one in horizontal and the other in vertical. Further
it was extended by associating a finite set of rules in the second
phase of generation with each table having either right linear
non-terminal rules of the form $A\rightarrow B$ or right-linear
terminal rules of the form $A\rightarrow a$ and such array
languages are denoted by $TRML$ and $TCFML$ and we have a well
known result $TRML\subset TCFML,RML\subset CFML\subset TCFML$
(refer 11). Correspondingly it can be established for hexagonal
arrays. Here we compare $pCDCFHAGS$ with these classes.
\begin{theorem}$pCD_3(CFHA,t)\backslash TCFML\ne\phi.$
\end{theorem}
\begin{proof}
Consider the $pCDCFHAGS$
\\
$G_5=(\{S,A,B,C,A',B',C',X,Y,Z\},\{a,b,c\},p_1,p_2,p_3,S)$\\
where the components are\\
\[
\begin{array}{l}
 P_1  = \left\{ {1)\begin{array}{*{20}c}
   S & {\begin{array}{*{20}c}
   \#   \\
   \#   \\
\end{array}} & \#   \\
\end{array} \to \begin{array}{*{20}c}
   X & {\begin{array}{*{20}c}
   A  \\
   B  \\
\end{array}}  \\
\end{array},2)\left( {\begin{array}{*{20}c}
   A & {\begin{array}{*{20}c}
   \#   \\
   {}  \\
\end{array}}  \\
\end{array} \to \begin{array}{*{20}c}
   Y & {\begin{array}{*{20}c}
   {A'}  \\
   {}  \\
\end{array}}  \\
\end{array},\left\{ B \right\}} \right),} \right. \\\\
\;\;\;\;\;\;\;\;\;\; 3)\left( {\begin{array}{*{20}c}
   B & {\begin{array}{*{20}c}
   {}  \\
   \#   \\
\end{array}}  \\
\end{array} \to \begin{array}{*{20}c}
   Y & {\begin{array}{*{20}c}
   {}  \\
   {B'}  \\
\end{array}}  \\
\end{array},\left\{ {A'} \right\}} \right),4)\left( {A' \to A,\left\{ {B'} \right\}} \right),5)\left( {B' \to B,\left\{ A \right\}} \right), \\\\
 \left.
\begin{array}{*{20}c}
   {}  \\
   {}  \\
\end{array}
 \;\;\;\;\;\;\;{6)\left( {A' \to Y,B' \to Y,Z' \to Z,C' \to C} \right)} \right\} \\
 \end{array}
 \]
\[
\begin{array}{l}
 P_2  = \left\{ {\begin{array}{*{20}c}
   Y & \#  & {\#  \to \begin{array}{*{20}c}
   a & b & {Z,\begin{array}{*{20}c}
   X & \#  & {\#  \to \begin{array}{*{20}c}
   a & a & {c,Z \to \begin{array}{*{20}c}
   b & {Z,}  \\
\end{array}}  \\
\end{array}}  \\
\end{array}}  \\
\end{array}}  \\
\end{array}} \right. \\
 \left. \;\;\;\;\;\;\;\;\;\;{C \to \begin{array}{*{20}c}
   a & {C'}  \\
\end{array}} \right\} \\\\
 P_3  = \left\{ {Z \to \begin{array}{*{20}c}
   X & {a,C \to a}  \\
\end{array}} \right\} \\
 \end{array}
\]

\[
\begin{array}{*{20}c}
   {} & {} & {} & {} & a & {} & b & {} & b & {} & a & {}  \\
   {} & {} & {} & a & {} & b & {} & b & {} & a & {} & {}  \\
   {} & {} & a & {} & b & {} & b & {} & a & {} & {} & {}  \\
   {} & a & {} & b & {} & b & {} & a & {} & {} & {} & {}  \\
   a & {} & a & {} & a & {} & a & {} & {} & {} & {} & {}  \\
   {} & a & {} & b & {} & b & {} & a & {} & {} & {} & {}  \\
   {} & {} & a & {} & b & {} & b & {} & a & {} & {} & {}  \\
   {} & {} & {} & a & {} & b & {} & b & {} & a & {} & {}  \\
   {} & {} & {} & {} & a & {} & b & {} & b & {} & a & {}  \\
\end{array}
\]
\centerline{Figure 6} Except the rules $Z^{'}\rightarrow Z ,
C^{'}\rightarrow C$ in $P_{1}$ generates the $LU$ and $LL$ arm
 with a middle marker $X$ and the symbols $Y$ above and below $X$ in both arms and which are equal in number.
 The first two rules of $P_{2}$ changes the symbols in the left most cell in to $a$.Then the remaining rules of $P_{2}$
 and the last two rules of $P_{1}$ ($Z^{,}\rightarrow Z$,$C^{,}\rightarrow C$) and the rules of the component $P_{3}$
 generate the arms such that each cell in the middle arm in the horizontal direction is made up of $a's$ and all other
 cells above and below are made up of $b's$ except the leftmost arms.The generation of the cells finally terminates,
 yielding the rightmost cells are rewritten by $a's$. Thus a hexagonal array in the shape of a left arrowhead
 (describing as Figure(6) is generated).If we treat $'b'$ as blank, such arrays cannot be in $TCFML$ and
 which in turn proves that $pCD_{3}(HCFA,t)\setminus TCFML \neq \phi$.
 \end{proof}

\noindent{\bf Conclusion.} In this paper, the picture array
generating power of cooperating $CF$ hexagonal array grammar
systems endowed with permitting symbols are studied. It is seen
that the control mechanism which we here used namely the
'permitting symbols' is shape preserving in picture generation and
also it reduce size complexity.

\end{document}